\documentclass[10pt]{book}
\usepackage{array}
\usepackage{amsthm,amsmath,amssymb}
\usepackage{graphicx}
\usepackage{listings}
\RequirePackage{natbib}
\RequirePackage[colorlinks,citecolor=blue,urlcolor=blue]{hyperref}

\theoremstyle{plain}
\newtheorem{lemma}{Lemma}
\newtheorem{theorem}{Theorem}

\newtheorem{proposition}{Proposition}
\renewcommand{\thefigure}{\arabic{figure}}

\begin{document}


\renewcommand{\baselinestretch}{1.2}


\markboth{\hfill{\footnotesize\rm WEIMING LI AND JIANFENG YAO} \hfill}
{\hfill {\footnotesize\rm LOCAL MOMENT ESTIMATION OF POPULATION SPECTRUM} \hfill}

\renewcommand{\thefootnote}{}
$\ $\par


\fontsize{10.95}{14pt plus.8pt minus .6pt}\selectfont
\vspace{0.8pc}
\centerline{\large\bf A LOCAL MOMENT ESTIMATOR OF THE SPECTRUM }
\vspace{2pt}
\centerline{\large\bf OF A LARGE DIMENSIONAL COVARIANCE MATRIX}
\vspace{.4cm}
\centerline{Weiming Li and Jianfeng Yao}
\vspace{.4cm}
\fontsize{9}{11.5pt plus.8pt minus .6pt}\selectfont

\footnote{
Weiming Li is Postdoctor, LMIB and School of Mathematics and Systems Science, Beihang University, Beijing, 100191, China. (E-mail: liwm601@gmail.com); and
Jianfeng Yao is Associate Professor, Department of Statistics and Actuarial Sciences,
The University of Hong Kong, Hongkong, China. (E-mail: jeffyao@hku.hk)
}

\begin{quotation}
\noindent {\it Abstract:}
This paper considers the problem of estimating the population spectral distribution from a sample covariance matrix in large dimensional situations. We generalize the contour-integral based method in \cite{M08a} and present a local moment estimation procedure. Compared with the original one, the new procedure can be applied successfully to models where the asymptotic clusters of sample eigenvalues generated by different population eigenvalues are not all separate. The proposed estimates are proved to be consistent. Numerical results illustrate the implementation of the estimation procedure and demonstrate its efficiency in various cases.\par

\vspace{9pt}
\par
\end{quotation}\par


\fontsize{10.95}{14pt plus.8pt minus .6pt}\selectfont

\setcounter{chapter}{1}
\setcounter{equation}{0} 
\noindent {\bf 1. Introduction}

\noindent

Let ${\bf x}_{1},\dots,{\bf x}_{n}$ be a sequence of i.i.d. zero-mean random vectors in $\mathbb{R}^{p}$ or $\mathbb{C}^p$, with a common  population covariance  matrix $\Sigma_{p}$. When the population size $p$ is not negligible with respect to the sample size $n$,
modern random matrix theory indicates that the sample covariance matrix $S_n=\sum_{i=1}^n{\bf x}_i{\bf x}_i^*/n$
does not approach $\Sigma_p$. Therefore, classical statistical procedures based on an approximation of $\Sigma_p$ by $S_n$ become inconsistent  in  such  large dimensional situations.

More precisely, the  {\em spectral   distribution} (SD) $F^A$ of an  $m\times  m$ Hermitian matrix (or real symmetric)  $A$ is  the measure generated by its eigenvalues $\{\lambda^A_i \}$,
\[    F^{A} = \frac1m \sum_{i=1}^m \delta_{\lambda^A_i}~,
 \]  where $\delta_b$ denotes the Dirac point measure at $b$. Denote by $(\sigma_{i})_{ 1\le i\le p}$ the $p$ eigenvalues of $\Sigma_p$.  We are particularly  interested  in the following  SD
 \[  H_p := F^{\Sigma_p}= \frac{1}{p} \sum_{i=1}^p \delta_{\sigma_i}.
 \]
In large dimensional frameworks, both dimensions $p$ and $n$ will grow to infinity. It is then natural to assume that $H_p$ converges weakly to a limit $H$. Both the SD $H_p$ and its limit $H$ are referred as the  {\em population spectral distribution}  (PSD) of the observation  model.

The main observation is that for large dimensional data, the empirical SD (ESD) $F_n:=F^{S_n}$ of $S_n$ is far from the PSD $H_p$. Indeed, under reasonable assumptions, when both dimensions $p$ and  $n$ grow proportionally,  almost surely, the ESD $F_n$ will weakly converge to a deterministic distribution $F$, which in general has no explicit form but is linked to the PSD $H$ via the so-called Mar\v{c}enko-Pastur equation, see \cite{MP67,Silverstein95,SilversteinB95}, and Section 2.1.

A natural question  here is the  recovering of the PSD $H_p$ (or its limit $H$)
from the ESD $F_n$. This question has a central importance in
several popular statistical methodologies like principal component
analysis \citep{Johnstone01}
or factor  analysis that
all rely on efficient estimations of some population covariance matrices.

Recent works on this problem include \cite{KarE08}, where the author proposed a nonparametric approach by solving the Mar\v{c}enko-Pastur equation on the upper complex plane, and then obtained consistent estimates of $H$. \cite{RaoJ08} investigated the asymptotic distributions of the moments of the ESD $F_n$ and introduced a Gaussian likelihood to get consistent estimates of $H$. In the work of \cite{M08a}, each mass of a discrete PSD $H$ is represented by a contour integral under a certain eigenvalue splitting condition and consistent estimators of $H$ are then obtained. Recently, \cite{Bai2010} modified the approach in \cite{RaoJ08} and turned it to a fully moments based procedure. Moreover beyond consistency, the authors proved also  a central limit theorem for the estimator. \cite{Li11} synthesized both the optimization approach in \cite{KarE08} and the parametric setup in \cite{Bai2010}, where an important improvement is that they changed the optimization problem from the complex plane to the real line by considering the extension of the Stieltjes transform on the real line.

Among all the above contributions, the contour-integral based method in \cite{M08a} is well known for its high efficiency and easy computation. However, the method is limited to a small class of discrete PSDs where, in addition, the imposed eigenvalue splitting condition states that distinct population eigenvalues should generate non-overlapping clusters of sample eigenvalues. Note that this method has been recently employed for subspace estimation in a so-called ``information plus noise" model in \cite{Hachem11}.

Our purpose in this paper is to extend Mestre's method to a more general situation where the splitting condition may not be satisfied. For a discrete PSD $H$ with finite support on $\mathbb R^+$, it is always true that one separate interval of the support $S_F$ of the limiting SD (LSD) $F$ corresponds to only one atom of $H$ if the dimension ratio $c$ is close to zero (the splitting condition holds). When $c$ is increased gradually, adjacent intervals of $S_F$ become closer, and some of them may ultimately merge into a larger interval (the splitting condition fails). Such merged intervals thus corresponds to more than one atom of $H$, and establishing their relationship in such a situation gives birth to our local estimation method.

Our strategy is that we first divide the PSD $H$ into a number of sub-probability measures, $H_1,\ldots,H_m$, such that each $H_i$ corresponds to one separate interval of $S_F$. Then, we develop a method to approximate the moments of $H_i$.  An estimate of $H_i$ can be obtained by solving a system of moment equations. Collecting all these estimates finally produces an estimator of $H$. It will be shown that when $m$ is equal to the number of atoms of $H$ (no merged intervals at all), this estimator reduces to the one in \cite{M08a}; If in contrary $m=1$ (all intervals merged into a single one), the estimator is equivalent to the one in \cite{Bai2010}.

The rest of the paper is organized as follows. In the next section, we review some useful results from Random Matrix Theory and introduce the division of a PSD $H$ according to the separation of the corresponding LSD $F$. A fast algorithm to solve the associated moment equations is also given. In Section 3, we present the theoretical supports and the detailed procedure of our estimation. In Section 4, simulation experiments are carried out to compare our new estimator with the estimator in \cite{M08a} and the moment estimator in \cite{Bai2010}. Some conclusions and remarks are presented in Section 5.

\par

\vskip 0.4cm
\setcounter{chapter}{2}
\setcounter{equation}{0} 
\noindent {\bf 2. Limiting spectral distribution and division of a PSD $H$}

\vskip 0.4cm
\noindent {\bf 2.1 The Mar\v{c}enko-Pastur equation}

Recall that the Stieltjes transform of $G$, a measure supported on the real line, is defined as $$s_G(z)=\int\frac{1}{x-z}dG(x),\quad z\in \mathbb{C}^{+},$$ where $\mathbb{C}^{+}$ is the set of complex numbers with positive imaginary part.

Let $S_G$ be the support set of $G$ and $S_G^c$ its complementary set. For the developments in this paper, we need to extend the Stieltjes transform to $\mathbb C\setminus S_G$ by
$$
s(z)=
\begin{cases}
s^*(z^*)         &(z\in\mathbb C^-=\{z\in\mathbb C: \Im(z)<0\}),\\
\lim_{\varepsilon\rightarrow0^+}s(x+\varepsilon\rm{i})&(z=x\in\mathbb R\setminus S_G),
\end{cases}
$$
where $a^*$ denotes the complex conjugate of $a$. The existence of the limit in the second term follows from the dominated convergence theorem.

Denote by $\lambda_1\leq\cdots\leq\lambda_p$ the eigenvalues of the sample covariance matrix $S_n$. Then the ESD $F_n$ of $S_n$ is
$$F_n=\frac{1}{p}\sum_{i=1}^p\delta_{\lambda_i},$$
whose Stieltjes transform is
\begin{eqnarray*}
s_n(z)=\int\frac{1}{x-z}dF_n(x)
=\frac{1}{p}\sum_{i=1}^p\frac{1}{\lambda_i-z}.
\end{eqnarray*}

Next, we present a convergence result of $F_n$ in \cite{Silverstein95} which is the basis of our estimation method in the next section.

\begin{lemma}\label{lema1}
Suppose that the entries of $X_n (p\times n)$ are complex random variables which are independent for each $n$ and identically distributed for all $n$, and satisfy ${\rm E}(x_{11})=0$ and ${\rm E}(|x_{11}|^2)=1.$ Also, assume that $T_n$ is a $p\times p$ random Hermitian nonnegative definite matrix, independent of $X_n$, and the empirical distribution $F^{T_n}$ converges almost surely to a probability measure $H$ on $[0,\infty)$ as $n\rightarrow\infty$. Set $B_n=T_n^{1/2}X_nX_n^{*}T_n^{1/2}/n.$ When $p=p(n)$ with $p/n\rightarrow c>0 $ as $n\rightarrow\infty$, then, almost surely, the empirical spectral distribution $F^{\rm B_n}$ converges in distribution, as $n\rightarrow\infty,$ to a (non-random) probability measure $F$, whose Stieltjes transform $s=s(z)$ is a solution to the equation
\begin{eqnarray}
s=\int\frac{1}{t(1-c-czs)-z}dH(t).\label{equ2.1}
\end{eqnarray}
The solution is also unique in the set $\{s\in {\mathbb C}: -(1-c)/z+cs\in{\mathbb C^+}\}.$
\end{lemma}

It will be more convenient to use a companion distribution $\underline{F}_n=(1-p/n)\delta_0+(p/n)F_n$ with Stiletjes transform
\begin{equation*}
\underline{s}_n(z)=-\frac{1-p/n}{z}+\frac{p}{n}s_n(z)=-\frac{1-p/n}{z}+\frac{1}{n}\sum_{i=1}^p\frac{1}{\lambda_i-z}.
\end{equation*}
The corresponding limit is $\underline{s}(z)=-(1-c)/z+cs(z)$ and it satisfies the following important equation which is a variant of Equation \eqref{equ2.1},
\begin{eqnarray}
z=-\frac{1}{\underline{s}}+c\int\frac{t}{1+t\underline{s}}dH(t).\label{equ2.3}
\end{eqnarray}
Both Equation \eqref{equ2.1} and Equation \eqref{equ2.3} are referred as the Mar\v{c}enko-Pastur equation.

Since the convergence in distribution of probability measures implies the pointwise convergence of the associated Stieltjes transforms, by Lemma \ref{lema1}, $\underline{s}_n(z)$ converges to $\underline{s}(z)$ almost surely, for any $z\in\mathbb C\setminus\mathbb R$. In \cite{SilversteinC95}, the convergence is extended to $S_F\setminus\{0\}$, and thus we conclude that for sufficiently large $n$, $\underline{s}_n(z)$ converges to $\underline{s}(z)$ almost surely for every $z\in\mathbb C\setminus (S_F\cup\{0\})$.

\vskip 0.4cm
\noindent {\bf 2.2 Division of a PSD $H$}

As mentioned in Introduction, our new method relies on a division of a PSD $H$ according to the separation of the corresponding LSD $F$. Suppose that the support $S_F$ of $F$ consists of $m$ ($m\geq1$) disjoint compact intervals, $S_1=[x_1^-,x_1^+],\ldots,S_m=[x_m^-,x_m^+]$  sorted in an increasing order.
Choose $\delta_i^-,\delta_i^+$ ($i=1,\ldots,m$) satisfying
\begin{equation}\label{equ2.4}
\delta_1^-<x_1^-<x_1^+<\delta_1^+<\delta_2^-<\cdots<\delta_{m-1}^+<\delta_m^-<x_m^-<x_m^+<\delta_m^+.
\end{equation}
Notice that when $z=x$ is restricted to $S_F^c$, $u(x)=-1/\underline{s}(x)$ is monotonically increasing and takes values in $S_H^c$ \citep{SilversteinC95}.
We have then
$$u(\delta_1^-)<u(\delta_1^+)<u(\delta_2^-)
<\cdots<u(\delta_{m-1}^+)<u(\delta_m^-)<u(\delta_m^+)$$
and
$$S_H\subset\bigcup_{i=1}^m\left[u(\delta_i^-),u(\delta_i^+)\right].$$
Consequently, we can match each compact interval of $S_F$ with a disjoint part of $S_H$ by
\begin{equation}\label{equ2.5}
S_i\rightarrow S_H\cap[u(\delta_i^-),u(\delta_i^+)], \quad i=1,\ldots,m,
\end{equation}
and hence, the PSD $H$ admits a division as follows:
$$H_i(A)=\int_{[u(\delta_i^-), u(\delta_i^+)]\cap A}dH,\quad A\in\mathcal B,\quad i=1,\ldots,m,$$
where $\mathcal B$ is the class of Borel sets of $\mathbb R$. Obviously, $\sum_{i=1}^mH_i=H$.

The map in \eqref{equ2.5} can be easily found out from the graph of $u(x)$ on $S_F^c$. Two typical representations of the graph are shown in Figure \ref{fig1}. The figures show that when $c<1$, each compact interval of $S_F$ corresponds to masses of $H$ that fall within this interval. But this is not true when $c>1$ as shown in the right panel of Figure \ref{fig1} where the mass 1 falls outside the interval $[x_1^-,x_1^+]$.

\begin{figure}
\begin{center}
\begin{minipage}[t]{0.5\linewidth}
\includegraphics[width=2.5in,height=1.8in]{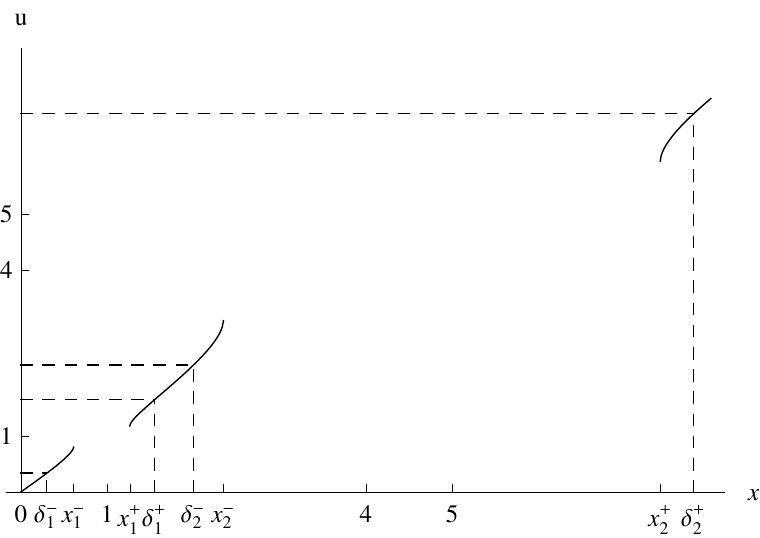}
\end{minipage}%
\begin{minipage}[t]{0.5\linewidth}
\includegraphics[width=2.5in,height=1.8in]{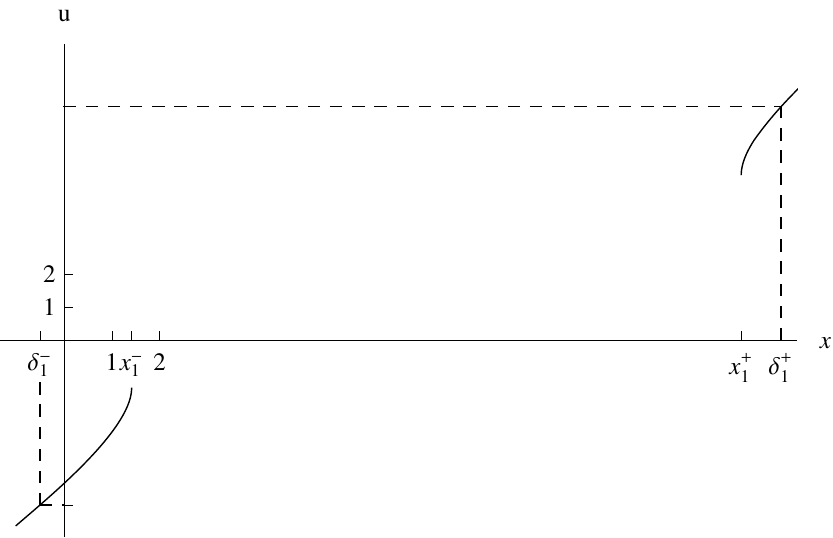}
\end{minipage}
\caption{The curves of $u(x)$ on $S_F^c\cap\mathbb R^+$ with $H_1=0.3\delta_1+0.4\delta_4+0.3\delta_5$ and $c_1=0.1$ (left), and $H_2=0.5\delta_1+0.5\delta_2$ and $c_2=4$ (right).}
\label{fig1}
\end{center}
\end{figure}

\vskip 0.4cm
\noindent {\bf 2.3 Moments of a discrete measure}

Let be a discrete measure $G=\sum_{i=1}^km_i\delta_{b_i}$ where $b_1<\cdots<b_k$ are $k$ masses with respective positive  weights $\{m_i\}$. Here, we don't assume $\sum m_i=1$ and $G$ can be a sub-probability measure. Define the $l$-th moment of $G$ as
$$\gamma_{l}=\sum_{i=1}^km_ib_i^l,\quad l=0,1,\ldots,$$
and the $N$-th Hankel matrix related to $G$ as
\begin{eqnarray*}
\Gamma(G,N)= \left(\begin{matrix}
\gamma_0&\gamma_1&\cdots&\gamma_{N-1}\\
\gamma_1&\gamma_2&\cdots&\gamma_{N}\\
\vdots&\vdots&&\vdots\\
\gamma_{N-1}&\gamma_{N}&\cdots&\gamma_{2N-2}
\end{matrix}\right).
\end{eqnarray*}

\begin{proposition}\label{pro1}
The Hankel matrix $\Gamma(G,k)$ is positive definite, and its determinant is
\begin{equation}
\det(\Gamma(G,k))=\prod_{i=1}^km_i\prod_{ 1\leq i<j\leq k}(b_i-b_j)^2.\label{equ2.6}\\
\end{equation}
Furthermore,
\begin{equation}
\det(\Gamma(G,N))=0,\quad N>k\label{equ2.7}.
\end{equation}
\end{proposition}
\begin{proof}
Write $M=\text{diag}(m_1,\ldots,m_k)$ a diagonal matrix, and
\begin{eqnarray*}
B= \left(\begin{matrix}
1&1&\cdots&1\\
b_1&b_2&\cdots&b_k\\
\vdots&\vdots&&\vdots\\
b_1^{k-1}&b_2^{k-1}&\cdots&b_k^{k-1}
\end{matrix}\right)
\end{eqnarray*}
which is a square Vandermonde matrix whose determinant is well known to be $\prod_{ 1\leq i<j\leq k}(b_j-b_i)$. From this and the fact that $\Gamma(G,k)=BMB^{T}$, we get Equation \eqref{equ2.6}.

Based on the above conclusion, Equation \eqref{equ2.7} and the positive definiteness of $\Gamma(G,k)$ can be verified by a direct calculation.
\end{proof}

Our aim here is to find an efficient inversion formula to these moment equations and the formula will be on the basis of our inference procedure below. Define a degree-$k$ polynomial $P(x)$ as
\begin{eqnarray*}
P(x)=\prod_{i=1}^k(x-b_i)=\sum_{i=0}^{k}c_ix^i,\quad c_{k}=1.
\end{eqnarray*}
Then, the coefficients $c_i$'s of $P(x)$ and the moments $\gamma_i$'s of $G$ have the following relationship.

\begin{proposition}\label{pro2}
Let ${\bf c}= (c_0,\ldots, c_{k-1})^\prime$ and ${\boldsymbol\gamma}= (\gamma_{k},\ldots, \gamma_{2k-1})^\prime$. Then,
$$\Gamma(G,k)\cdot\bf c+{\boldsymbol\gamma}=0.$$
\end{proposition}
\begin{proof}
It is easily verified.
\end{proof}

Propositions \ref{pro1} and \ref{pro2} establish a one-to-one map between the parameters of $G$ and its moments. They further tell us that the masses of $G$ are all zeros of $P(x)$ with coefficients ${\bf c}=-(\Gamma(G,k))^{-1}\cdot{\boldsymbol\gamma}$ and $c_{k}=1$.
As to the weights of $G$, they can be trivially obtained by solving linear equations,
$$\sum_{i=1}^km_ib_i^l=\gamma_l,\quad l=0,\ldots,k-1,$$
with $b_i$'s known.

\par

\par
\vskip 0.4cm
\setcounter{chapter}{3}
\setcounter{equation}{0} 
\noindent {\bf 3. Estimation}

\vskip 0.4cm
\noindent {\bf 3.1 Model and estimation strategy}

We consider a class of discrete PSDs with finite support on $\mathbb R^+$, that is,
\begin{equation*}
H({\boldsymbol\theta})=w_1\delta_{a_1}+\cdots+ w_k\delta_{a_k},\quad {\boldsymbol\theta}\in{\Theta},
\end{equation*} where
\begin{eqnarray*}
{\Theta} =\bigg\{{\boldsymbol\theta}=(a_1, w_1, \ldots, a_k, w_k): 0< a_1<\cdots <a_k<\infty;\ w_i>0,\ \sum_{i=1}^{k}w_i=1\bigg\}.
\end{eqnarray*}
Here, the order $k$ of $H$ is assumed known (when $k$ is also to be estimated, a consistent estimator of $k$ is given in \cite{CBY11}).

Suppose that the support $S_F$ of the LSD $F$ associated to $H$ and $c$ has $m$ ($1\leq m\leq k$) disjoint compact intervals. According to the discussion in Section 2, $H$ can be divided into $m$ parts, $H_1,\ldots, H_m$, with $H_i$ consisting of $k_i$ masses of $H$, $k_i\geq1$ and $\sum_{i=1}^mk_i=k$.

When $k_i$'s are all known and equal to 1, the assumption reduces to the split case in \cite{M08a}. By contrast, we consider that $k_i$'s are unknown, can be larger than 1, and are not necessarily equal.

Our estimation strategy is the following:
\begin{itemize}
\item[1)] determine the division of $H$ according to the separation of clusters of sample eigenvalues;
\item[2)] for each part $H_i$, obtain strongly consistent estimators of its moments;
\item[3)] obtain a strongly consistent estimator $\widehat{{\bf k}}_n$ of the partition $(k_1,\ldots,k_m)$ of numbers of masses in the $m$ parts $H_1,\ldots,H_m$;
\item[4)] by combination of these estimators and using the method of moments, finally obtain consistent estimators of all the weights and masses $(w_i,a_i)$.
\end{itemize}

Note that in the first step, an accurate division of $H$ may not be always achieved, especially when sample sizes are relatively small. A solution to this problem will be given later.

\vskip 0.4cm
\noindent {\bf 3.2 Estimation of the moments of $H_i$}

The following theorem re-expresses the moments of $H_i$ by contour integrals related to the companion Stieltjes transform $\underline{s}(z)$.

\begin{theorem}\label{th1} Suppose the assumptions in Lemma \ref{lema1} are fulfilled, then the $l$-th moment of $H_i$ can be expressed as
\begin{equation}\label{eq3.2}
\gamma_{i,l}=(-1)^{l}\frac{1}{c}\frac{1}{2\pi \rm i}\oint_{C_i}\frac{z\underline{s}^\prime(z)}{\underline{s}^l(z)}dz,\quad l=1,2,\ldots,
\end{equation}
where $C_i$ is a positively oriented contour described by the boundary of the rectangle
$$\{ z\in\mathbb C: \delta_i^-\leq\Re(z)\leq\delta_i^+, |\Im(z)|\leq1\},$$
where $\delta_i^-, \delta_i^+$ $(i=1,\ldots,m)$ are defined by \eqref{equ2.4} and $\delta_1^-<0$ if $c\geq 1$.
\end{theorem}

\begin{proof}
Let the image of $C_i$ under $u(z)=1/\underline{s}(z)$ be $$u(C_i)=\{u(z):z\in C_i\}.$$
Notice that $\underline{s}(z)$ is holomorphic on $C_i$. Then, $u(C_i)$ is a simple closed curve taking values on $\mathbb C\setminus (S_H\cup\{0\})$. (The function $u(z)=-1/\underline{s}(z)$ is analytic on $C_i$ and is a one-to-one map from $C_i$ to its image $u(C_i)$. Thus, the two curves $C_i$ and $u(C_i)$ are homeomorphic. Since $C_i$ is simple and closed (homeomorphic to a unit circle in $\mathbb C$), its image is also simple and closed.) Moreover, since $\Im(u(z))\neq0$ for all $z$ with $\Im(z)\neq0$, we have $u(C_i)\cap\mathbb R=\{u(\delta_i^-),u(\delta_i^+)\}$ and $u(C_i)$ encloses $[u(\delta_i^-),u(\delta_i^+)]$. Therefore, $u(C_i)$ encloses only $S_{H_i}$ and no other masses of $H$.

Applying this change of variable to the right hand side of \eqref{eq3.2}, we have
\begin{eqnarray*}
(-1)^{l}\frac{1}{c}\frac{1}{2\pi \rm i}\oint_{C_i}\frac{z\underline{s}^{\prime}(z)}{\underline{s}^l(z)}dz
&=&\frac{1}{c}\frac{1}{2\pi \rm i}\oint_{u(C_i)}z(u)u^{l-2}du\\
&=&\frac{1}{c}\frac{1}{2\pi \rm i}\oint_{u(C_i)}u^{l-1}+c\int\frac{tu^{l-1}}{u-t}dH(t)du\\
&=&\frac{1}{2\pi \rm i}\int\oint_{u(C_i)}\frac{tu^{l-1}}{u-t}dudH(t)\\
&=&\gamma_{i,l},
\end{eqnarray*}
where the second equation is from the Mar\v{c}enko-Pastur equation, and the last equation follows from the residue theorem.
\end{proof}

By substituting the empirical Stieltjes transform $\underline{s}_n(z)$ for $\underline{s}(z)$ in \eqref{eq3.2}, we get a natural estimator of ${\gamma}_{i,l}$:
\begin{equation}\label{eq3.3}
\widehat{\gamma}_{i,l}=(-1)^{l}\frac{n}{p}\frac{1}{2\pi \rm i}\oint_{C_i}\frac{z\underline{s}_n^\prime(z)}{\underline{s}_n^l(z)}dz,\quad l=1,2,\ldots.
\end{equation}

\begin{theorem}\label{th2}
Under the assumptions of Lemma \ref{lema1}, for each $l$ $(l\geq 1)$, $\widehat{\gamma}_{i,l}$ converges almost surely to $\gamma_{i,l}$.
\end{theorem}
\begin{proof} From the fact that for sufficiently large $n$, with probability one, there are no sample eigenvalues located outside $S_F$ \citep{BS98}, we have then, for sufficiently large $n$, $z\underline{s}_n^\prime(z)/\underline{s}_n^l(z)$ as well as $z\underline{s}^\prime(z)/\underline{s}^l(z)$ are continuous on $C_i$, and thus bounded on the contour. By the convergence of $\underline{s}_n(z)$ and the dominated convergence theorem, almost surely,
\begin{eqnarray*}
|\gamma_{i,l}-\widehat{\gamma}_{i,l}|&=&\bigg|\oint_{C_i}\frac{z\underline{s}^\prime(z)}{\underline{s}^l(z)}-
\frac{z\underline{s}_n^\prime(z)}{\underline{s}_n^l(z)}dz\bigg|\\
&\leq&\oint_{C_i}\bigg|\frac{z\underline{s}^\prime(z)}{\underline{s}^l(z)}-\frac{z\underline{s}_n^\prime(z)}{\underline{s}_n^l(z)}\bigg||dz|\\
&\rightarrow&0,\quad n\rightarrow\infty.
\end{eqnarray*}
\end{proof}

A technical issue here is the contour integration in \eqref{eq3.3}. It can be calculated by the residue theorem and an algorithm is described in Appendix.

\vskip 0.4cm
\noindent {\bf 3.3 Estimation of the partition $(k_1,\ldots,k_m)$}

Denote by ${\bf k}=(k_1,\ldots,k_m)'$ the vector of orders of $H_i$'s, the collection of all possible values of ${\bf k}$ is
$$\mathbb K=\{{\bf k}: k_i\geq1,\ \sum_{i=1}^mk_i=k\}.$$
Let ${\bf k}_0=(k_{0,1},\ldots,k_{0,m})'$ be the true value of ${\bf k}$. From Proposition \ref{pro1}, we know that the smallest eigenvalue $\lambda_{\min}(\Gamma(H_i,k_i))$ of the Hankel matrix $\Gamma(H_i,k_i)$ is positive if $k_i\leq k_{0,i}$, and otherwise 0. Based on this property, we construct the following objective function
\begin{equation*}
g({\bf k})=\min\left\{\lambda_{\min}(\Gamma(H_i,k_i)), \quad i=1,\ldots,m\right\},\quad{\bf k}\in \mathbb K,
\end{equation*}
that satisfies
$$g({\bf k}_0)>0 \quad {\rm and}\quad g({\bf k})=0\quad ({\bf k}\neq {\bf k}_0).$$
So, an estimator of ${\bf k}_0$ can be obtained by maximizing the estimate of $g({\bf k})$, i.e.
\begin{eqnarray*}
\widehat{{\bf k}}_n&=&\arg\max_{{\bf k}\in\mathbb K}\widehat{g}({\bf k})\\
&=&\arg\max_{{\bf k}\in\mathbb K}\min\left\{\lambda_{\min}(\widehat{\Gamma}(H_i,k_i)), \quad i=1,\ldots,m\right\},
\end{eqnarray*}
where  $\widehat{\Gamma}(H_i,k_i)=(\widehat{\gamma}_{i,r+s-2})_{1\leq r,\ s\leq k_i}$ with its entries defined by \eqref{eq3.3}.

Note that when evaluating the estimator $\widehat{{\bf k}}_n$, it is not necessary to compare $\widehat{g}({\bf k})$'s at all ${\bf k}$-points, but only at a small part of them. More precisely, for the $i$-th element $k_i$ of $\bf k$, in theory, its value may range from 1 to $k-m+1$ and its true value $k_{0,i}$ makes $\Gamma(H_i,k_i)$ positive definite. This implies that if $\Gamma(H_i,k_i)$ is non-positive definite then $k_i\neq k_{0,i}$ (actually $k_i>k_{0,i}$). Based on this knowledge, in practice, it is enough to consider $k_i$ that belongs to a set $\{1,\ldots, d_i\}$, where $d_i\leq k-m+1$ stands for the largest integer such that $\widehat{\Gamma}(H_i,d_i)$ is positive definite. This technique can effectively reduce the computational burden when the cardinality of $\mathbb K$ is large.

\begin{theorem}\label{th3}
Under the assumptions of Lemma \ref{lema1}, almost surely,
$$\widehat{{\bf k}}_n\rightarrow{\bf k}_0,\quad {\rm as}\quad n\rightarrow\infty.$$
\end{theorem}
\begin{proof}
The conclusion follows from Theorem \ref{th2} and the fact that ${\bf k}_0$ is the unique maximizer of the function $g({\bf k})$ on the finite set $\mathbb K$.
\end{proof}

\vskip 0.4cm
\noindent {\bf 3.4 Estimation of ${\boldsymbol\theta}$}

By Theorem \ref{th3} and since the partition set $\mathbb K$ is finite, almost surely, $\widehat{\bf k}_n={\bf k}_0$ eventually. As far as the consistency is concerned for estimation of $\boldsymbol \theta$, we may assume in this section that the partition $\bf k$ is known without loss of generality. Then, the estimator $\widehat{\boldsymbol\theta}_{n}$ of ${\boldsymbol\theta}$ is defined to be a solution of the following $2k$ equations:
\begin{eqnarray}\label{eq3.5}
\int x^ldH_i({\boldsymbol\theta})=\widehat{\gamma}_{i,l},\quad l=0,\ldots,2k_i-1,\quad i=1,\ldots,m,
\end{eqnarray}
where $\widehat{\gamma}_{i,0}=v_i/v$, $i=1,\ldots,m$, ($v$ is the total number of positive sample eigenvalues and $v_i$ is the number of those forming the $i$-th cluster). We call $\widehat{{\boldsymbol\theta}}_{n}$ the {\em local moment estimator} (LME) of $\boldsymbol\theta$, since it is obtained by the moments of $H_i$'s, rather than the moments of $H$. Accordingly, the LME of $H$ is $\widehat{H}=H(\widehat{\boldsymbol\theta}_n)$. When $k_1=\cdots=k_m=1$, the LME reduces to the one in \cite{M08a}.

The solution of the moment equations \eqref{eq3.5} exists and is unique if the matrices $\widehat{\Gamma}(H_i,k_i)$'s are all positive definite. Moreover, a fast algorithm for the solution exists following the equations given in Section 2.3: indeed, the algorithm needs to solve a one-variable polynomial equation and a linear system.


Next, we establish the strong consistency of the LME as follows.
\begin{theorem}\label{th4}
In addition to the assumptions in Lemma \ref{lema1}, suppose that the true value of the parameter vector ${\boldsymbol\theta}_0$ is an inner point of $\Theta$. Then, the LME $\widehat{\theta}_n$ is strongly consistent: almost surely,
$$\widehat{{\boldsymbol\theta}}_n\rightarrow {\boldsymbol\theta}_0,\quad n\rightarrow\infty.$$
\end{theorem}

\begin{proof}
Write ${\widehat{\boldsymbol\theta}_n}=({\widehat{\boldsymbol\theta}_{1n}},\ldots,{\widehat{\boldsymbol\theta}_{mn}})$, where $\widehat{\boldsymbol\theta}_{in}$ is the LME of the parameter vector ${\boldsymbol\theta}_{i0}$ of $H_i$ ($i=1,\ldots,m$). It is sufficient to prove that, almost surely,
$$\widehat{\boldsymbol\theta}_{in}\rightarrow\boldsymbol\theta_{i0},\quad n\rightarrow\infty,$$
for each $i$ ($i=1,\ldots,m$).

Let $h_i$ be the function  $R^{2k_i} \to R^{2k_i}$:
\[    {\boldsymbol\theta}_i \mapsto {\boldsymbol\gamma}_{i}=\left(\gamma_{i,0},\ldots,\gamma_{i,2k_i-1}\right).
\]
Then the multivariate function $h_i$ is invertible from the conclusions of Propositions \ref{pro1} and \ref{pro2}.

Denote $\widehat{\boldsymbol\gamma}_{in} = \left(\widehat\gamma_{i,0},\ldots,\widehat\gamma_{i,2k_i-1}\right)$
and $\boldsymbol\gamma_{i0} = h_i(\boldsymbol\theta_{i0})$. By the convergence of $\widehat{\boldsymbol\gamma}_{in}$ (Theorem \ref{th2}) and the implicit function theorem, there exists a
neighborhood $U_i$ of $\boldsymbol\theta_{i0}$ and a neighborhood $V_i$ of $\boldsymbol\gamma_{i0}$, such that $h_i$ is a differomorphism from $U_i$ onto $V_i$. Moreover, $\widehat{\boldsymbol\theta}_{in} = h_i^{-1}(\widehat{\boldsymbol\gamma}_{in})\in U_i$ exists almost surely for all large $n$.
Therefore, $\widehat{\boldsymbol\theta}_{in}$ converges to $\boldsymbol\theta_{i0}=h_i^{-1}(\boldsymbol\gamma_{i0})$
almost surely, as $n\rightarrow\infty$.
\end{proof}

\vskip 0.4cm
\noindent {\bf 3.5 A generalization of the local moment estimator}

The proposed estimation procedure needs a good judgment on the separation of clusters of sample eigenvalues. This may be indeed a problem when two or more adjacent clusters are very close, which can happen when the sample size is too small. To handle this problem, we introduce here a generalized version of the estimation procedure. The resulting estimator is referred as generalized LME (GLME).

Suppose that the support $S_F$ has $m\ (\geq1)$ disjoint compact intervals, and accordingly $H$ gains a division of $m$ parts: $H_1,\ldots, H_m$. Without loss generality, we suppose that the first two clusters of sample eigenvalues have no clear separation under a situation of finite sample size. Our strategy to cope with this is simply to merge these two clusters into one and treat $H_1$ and $H_2$ as a whole. Then, the GMLE can be obtained by conducting a similar procedure of estimation as mentioned in Section 3.1.

An extreme case of the GLME is to merge all clusters into one, then one may find with surprise that the GLME becomes a ``full moment" estimator which is equivalent to the moment estimator in \cite{Bai2010}. In this sense, the GLME encompasses this moment method. However, the merging procedure may result in a reduction of estimation efficiency, which will be illustrated numerically in the next section.

On theoretical aspect, it can be easily shown that Theorems \ref{th1}--\ref{th3} still hold true after the merging procedure. We can therefore obtain the strong convergence of the GLME by a similar proof of Theorem \ref{th4}. Hence these proofs are omitted.

\vskip 0.4cm
\setcounter{chapter}{4}
\setcounter{equation}{0} 
\noindent {\bf 4. Simulation}

In this section, simulations are carried out to examine the performance of the proposed estimator comparing with the estimator in \cite{M08a} (referred as ME), and the one in \cite{Bai2010} (referred as BCY).

Samples are drawn from mean-zero normal distribution with $(p, n)$ = (320, 1000) for the estimation of $H$, and $(p, n)$ = (320, 1000), (160, 500), (64, 200), (32, 100), (16, 50) for the estimation of the partition $\bf k$ of $H$. The independent replications are 1000. More $p/n$ combinations are considered for the partition estimator $\widehat{\bf k}_n$ since this step has a primary importance on the overall performance of the procedure.

In order to measure the distance between $H$ and its estimate $\widehat{H}$, we consider the Wasserstein distance $d=\int |Q_{H}(t)-Q_{\widehat{H}}(t)|dt$ where $Q_{\mu}(t)$ is the quantile  function of a probability measure $\mu$. Execution times are also provided for one realization of $\widehat{H}$ in seconds. All programs are realized in Mathematica 8 software, and run on a PC equipped with 3.5GHk CPU and 8GB physical RAM.

We first consider a case in \cite{M08a} where $H=0.5\delta_1+0.25\delta_7+0.125\delta_{15}+0.125\delta_{25}$ and $c=0.32$.
In this case, $H$ has four atoms at 1, 7, 15, and 25, while the sample eigenvalues form three clusters, and spread over $S_F=[0.2615, 1.6935]\cup[3.2610, 10.1562]\cup[10.2899, 38.0931]$ in the limit, see Figure \ref{fig3}. In Mestre's paper, it was shown that the ME performed very well by assuming all weight parameters (multiplicities) being known even if the splitting condition is not verified by the last two atoms.

In the viewpoint of the LME method, the PSD $H$ can only be divided into three parts: $H_1=0.5\delta_1$, $H_2=0.25\delta_7$, and $H_3=0.125\delta_{15}+0.125\delta_{25}$. Thus, the true partition of $H$ is ${\bf k}_0=(1,1,2)$. Table \ref{table1} presents the frequency of estimates of the partition ${\bf k}$. The results show that the true model can be identified with an accuracy of 100\% when the sample size $n$ is larger than 200, and the accuracy decreases as $n$ goes smaller.

Table \ref{table2} presents statistics for the three estimators of $H$. The first six rows are results assuming all the weights $\{w_i\}$ are known, while in the last four rows are results assuming only $\{w_1, w_2\}$ are known and $w_3$ is to be estimated ($w_4$ is determined by $\sum w_i=1$). Overall, the LME is as good as the ME when all weights are known, and is much better than the BCY in all cases. When $w_3$ is unknown, the problem is harder resulting larger distance values of $d$ for both methods LME and BCY. This difficulty is also reflected by larger variances of the estimates of $a_3$ and $a_4$ which are closely related to the parameter $w_3$ (and $w_4$). Concerning the execution time shown in the table, the BCY is the fastest followed by the ME, and then by the LME. However, the elapsed time of the BCY estimation increases rapidly as the number of unknown parameters increases.

\begin{figure}
\begin{minipage}[t]{0.5\linewidth}
\includegraphics[width=2.5in]{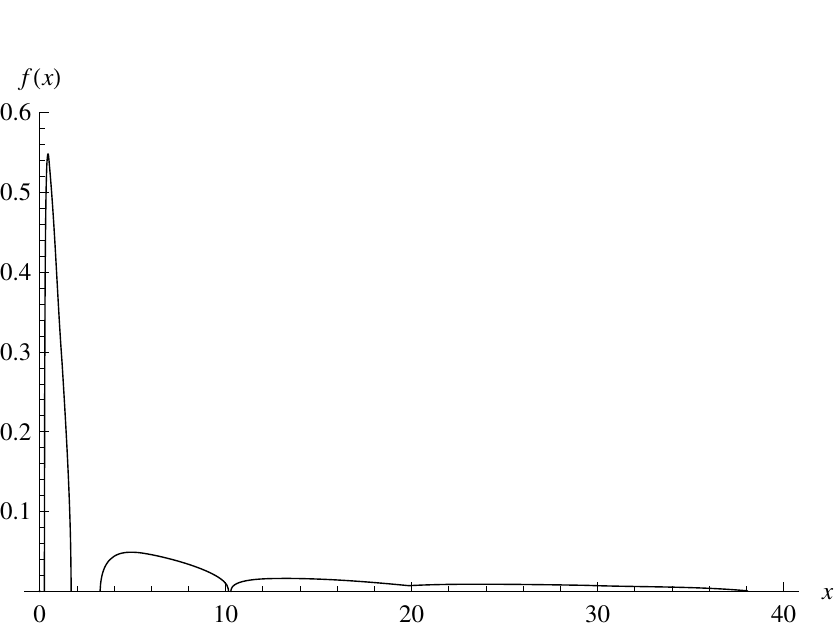}
\end{minipage}%
\begin{minipage}[t]{0.5\linewidth}
\includegraphics[width=2.5in]{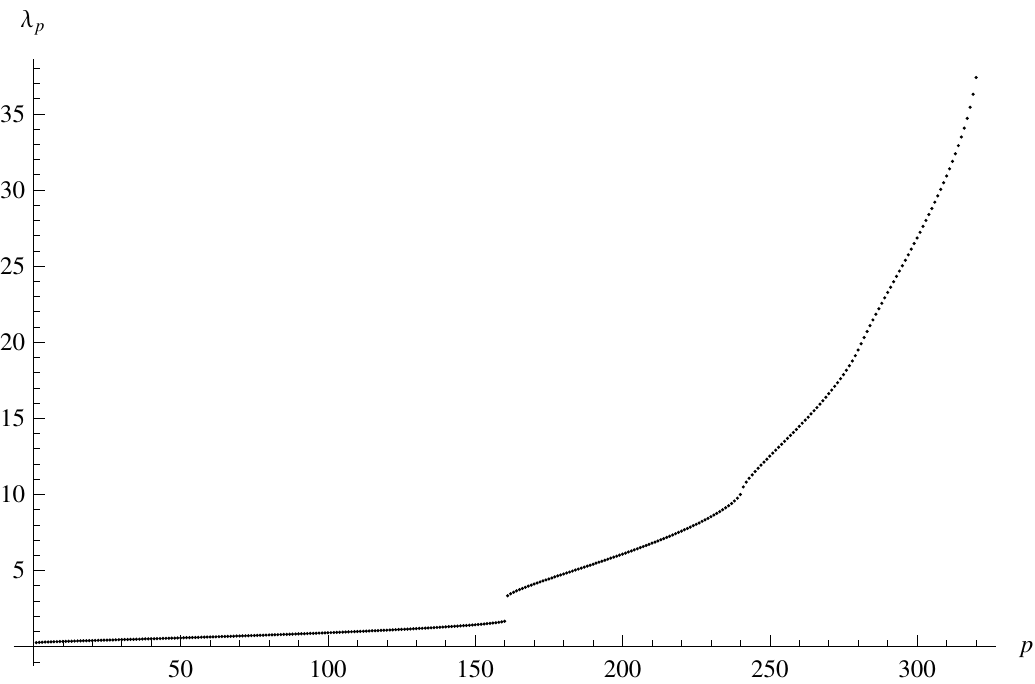}
\end{minipage}
\caption{The density curve of $F$ (left) and the average of the $i$-th ($i=1,\ldots,320$) sample eigenvalues (right) from 1000 replications for $H=0.5\delta_1+0.25\delta_7+0.125\delta_{15}+0.125\delta_{25}$ and $c=0.32$.}
\label{fig3}
\end{figure}

\begin{table}
\begin{center}
\caption{Frequency of estimates for the partition of $H$: $H_1=0.5\delta_1$, $H_2=0.25\delta_7$, $H_3=0.125\delta_{15}+0.125\delta_{25}$ with $p/n=0.32$.}
\label{table1}
\begin{tabular}{|lccc|}
\hline
Dimensions&{\bf k}=$(1,1,2)'$&{\bf k}=$(1,2,1)'$&{\bf k}=$(2,1,1)'$\\
\hline
$(p,n)=(320,1000)$&1000&0&0\\
$(p,n)=(160,500)$ &1000&0&0\\
$(p,n)=(64,200)$  &999 &0&1\\
$(p,n)=(32,100)$  &896&45&59\\
$(p,n)=(16,50)$   &623&169&208\\
\hline
\end{tabular}
\end{center}
\end{table}

\begin{table}
\begin{center}
\caption{Estimates for $H=0.5\delta_1+0.25\delta_7+0.125\delta_{15}+0.125\delta_{25}$ with $p=320$ and $n=1000$.}
\label{table2}
\begin{tabular}{|llccccccc|}
\hline
&&$a_1 $&$a_2$&$a_3$&$w_3$&$a_4$&$d$&Time\\
\hline
ME        &Mean   &1.0000&7.0031& 14.9987&-     & 25.0001&0.0425&0.533s\\
          &St. D. &0.0041&0.0407&\ 0.1368&-     &\ 0.1964&0.0199&\\
LME       &Mean   &1.0000&7.0060& 14.9533&-     & 25.0381&0.0447&0.578s\\
          &St. D. &0.0040&0.0401&\ 0.1371&-     &\ 0.2033&0.0205&\\
BCY       &Mean   &0.9924&7.0387& 14.8968&-     & 25.0658&0.0887&0.147s\\
          &St. D. &0.0189&0.1204&\ 0.3027&-     &\ 0.2312&0.0554&\\
LME$^*$   &Mean   &1.0000&7.0027& 14.9935&0.1259& 25.0772&0.1136&0.890s\\
          &St. D. &0.0040&0.0401&\ 0.2398&0.0059&\ 0.3520&0.0662&\\
BCY$^*$   &Mean   &1.0012&6.9806& 15.1350&0.1288& 25.1728&0.2143&0.710s\\
          &St. D. &0.0082&0.0753&\ 0.5738&0.0113&\ 0.4903&0.1368&\\
\hline
\end{tabular}
\end{center}
\end{table}

It should be noticed that in general when the splitting condition is not satisfied, the performance of the ME may decrease sharply, and the estimates may suffer from large biases. Next, we show this phenomenon and also examine the performances the LME and the BCY in such situations.

We consider a similar model where the third atom of $H$ is set to be 20 instead of 15 and other settings remain unchanged, that is, $H=0.5\delta_1+0.25\delta_7+0.125\delta_{20}+0.125\delta_{25}$ and $c=0.32$. The empirical and limiting distributions of sample eigenvalues are illustrated in Figure \ref{fig4}, where $S_F=[0.2617, 1.6951]\cup[3.2916, 10.4557]\cup[12.3253, 39.2608]$.

Analogous statistics are summarized in Tables \ref{table3} and \ref{table4}. The results in Table \ref{table3} show that the estimation of the partition $\bf k$ is more difficult in this case, but its accuracy still achieves 100\% with the sample size $n=1000$.  The statistics in Table \ref{table4} reveal that the estimators of $a_3$ and $a_4$ from the ME have a bias as large as 0.85 in average when all weight parameters are assumed known, while the LME and the BCY are unbiased in the same settings. On the other hand, it is again confirmed that the LME improves upon the BCY, especially when the weight parameters are partially unknown.

\begin{figure}
\begin{center}
\begin{minipage}[t]{0.5\linewidth}
\includegraphics[width=2.5in]{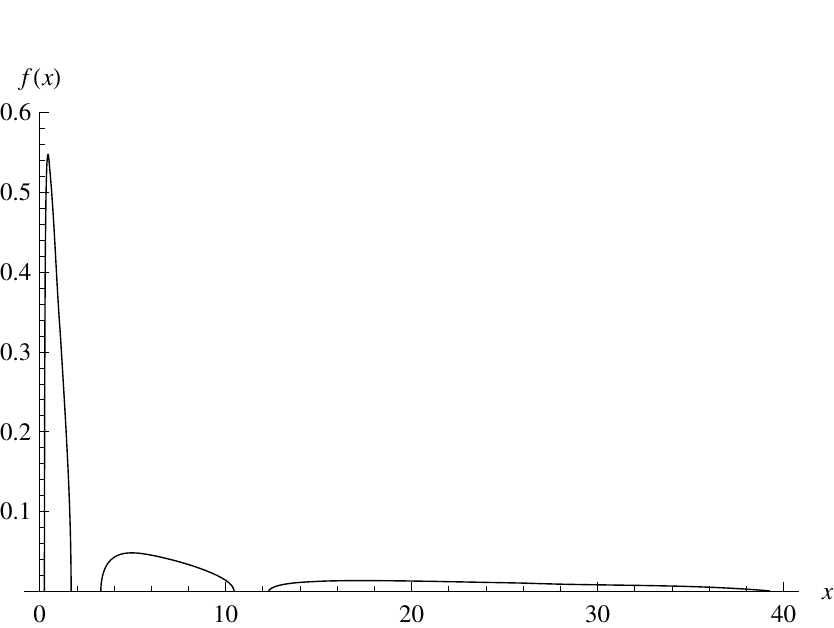}
\end{minipage}%
\begin{minipage}[t]{0.5\linewidth}
\includegraphics[width=2.5in]{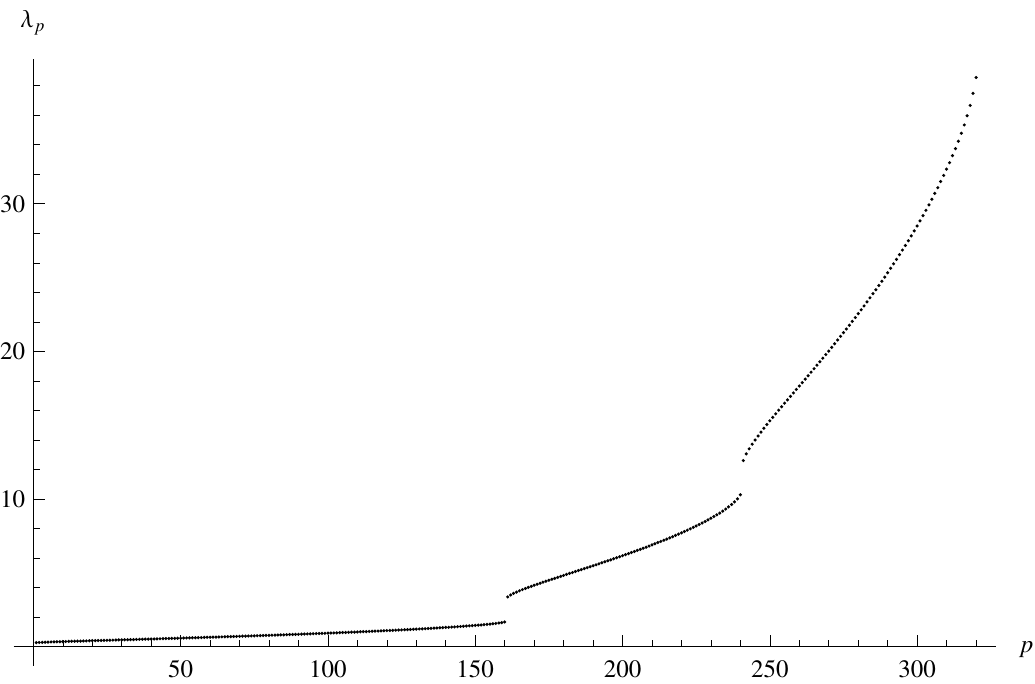}
\end{minipage}
\caption{The density curve of $F$ (left) and the average of the $i$-th ($i=1,\ldots,320$) sample eigenvalues (right) from 1000 replications for $H=0.5\delta_1+0.25\delta_7+0.125\delta_{20}+0.125\delta_{25}$ and $c=0.32$.}
\label{fig4}
\end{center}
\end{figure}

\begin{table}
\begin{center}
\caption{Frequency of estimates for the partition of $H$: $H_1=0.5\delta_1$, $H_2=0.25\delta_7$, $H_3=0.125\delta_{20}+0.125\delta_{25}$ with $p/n=0.32$.}
\label{table3}
\begin{tabular}{|lccc|}
\hline
Dimensions&{\bf k}=$(1,1,2)'$&{\bf k}=$(1,2,1)'$&{\bf k}=$(2,1,1)'$\\
\hline
$(p,n)=(320,1000)$&1000&0 &0\\
$(p,n)=(160,500)$ &922 &28&50\\
$(p,n)=(64,200)$  &595 &183&222\\
$(p,n)=(32,100)$  &455 &267&278\\
$(p,n)=(16,50)$   &376 &260&364\\
\hline
\end{tabular}
\end{center}
\end{table}

\begin{table}
\begin{center}
\caption{Estimates for $H=0.5\delta_1+0.25\delta_7+0.125\delta_{20}+0.125\delta_{25}$ with $p=320$ and $n=1000$.}
\label{table4}
\begin{tabular}{|llccccccc|}
\hline
&&$a_1 $&$a_2$&$a_3$&$w_3$&$a_4$&$d$&Time\\
\hline
ME        &Mean   &1.0001&6.9996&\bf{19.1483}&-&\bf{25.8521}&\bf{0.2224}&0.533s\\
          &St. D. &0.0041&0.0395&\ 0.1836&-     &\ 0.2068&0.0404&\\
LME       &Mean   &1.0000&7.0006& 19.9157&-     & 25.0811&0.0620&0.575s\\
          &St. D. &0.0040&0.0391&\ 0.2404&-     &\ 0.2631&0.0341&\\
BCY       &Mean   &0.9965&7.0090& 19.9028&-     & 25.0874&0.0875&0.142s\\
          &St. D. &0.0126&0.0692&\ 0.3456&-     &\ 0.3155&0.0516&\\
LME$^*$   &Mean   &1.0000&7.0003& 19.8739&0.1282& 25.2896&0.2588&0.896s\\
          &St. D. &0.0039&0.0390&\ 0.7883&0.0342&\ 0.8857&0.1464&\\
BCY$^*$   &Mean   &0.9993&6.9983& 19.8587&0.1331& 25.4569&0.3286&0.865s\\
          &St. D. &0.0054&0.0446&\ 1.2884&0.0437&\ 1.0888&0.1685&\\
\hline
\end{tabular}
\end{center}
\end{table}

Finally, we study a case where $H=0.5\delta_1+0.25\delta_3+0.125\delta_{15}+0.125\delta_{25}$ and $c=0.32$ to examine the performance of the GLME. The empirical and limiting distributions of sample eigenvalues are illustrated in Figure \ref{fig5}, where $S_F=[0.2552, 1.6086]\cup[1.6609, 4.7592]\cup[9.1912, 37.6300]$. With the used dimensions, the first two clusters of sample eigenvalues are too close to be identified, and we have to merge these two clusters into one to get the GLME of $H$ (thus no weight parameters are known at all). For comparison, we also present the LME by assuming that we know the true separation of $S_F$ into three intervals (which is not seen from the data).

Statistics in Table \ref{table5} show a perfect estimation of $\bf k$ with sample sizes $n=500, 1000$. Results in Table \ref{table6} demonstrate that the GMLE has a very good performance with only a slight reduction in estimation efficiency compared with the (impractical) LME.

Note that the BCY becomes unstable for this model as, for example, the empirical moment equations defining the estimator often have no real solutions. A major reason is that the required estimates of the 6-th and 7-th moments of $H$ have poor accuracy in such a situation.

\begin{figure}
\begin{center}
\begin{minipage}[t]{0.5\linewidth}
\includegraphics[width=2.5in]{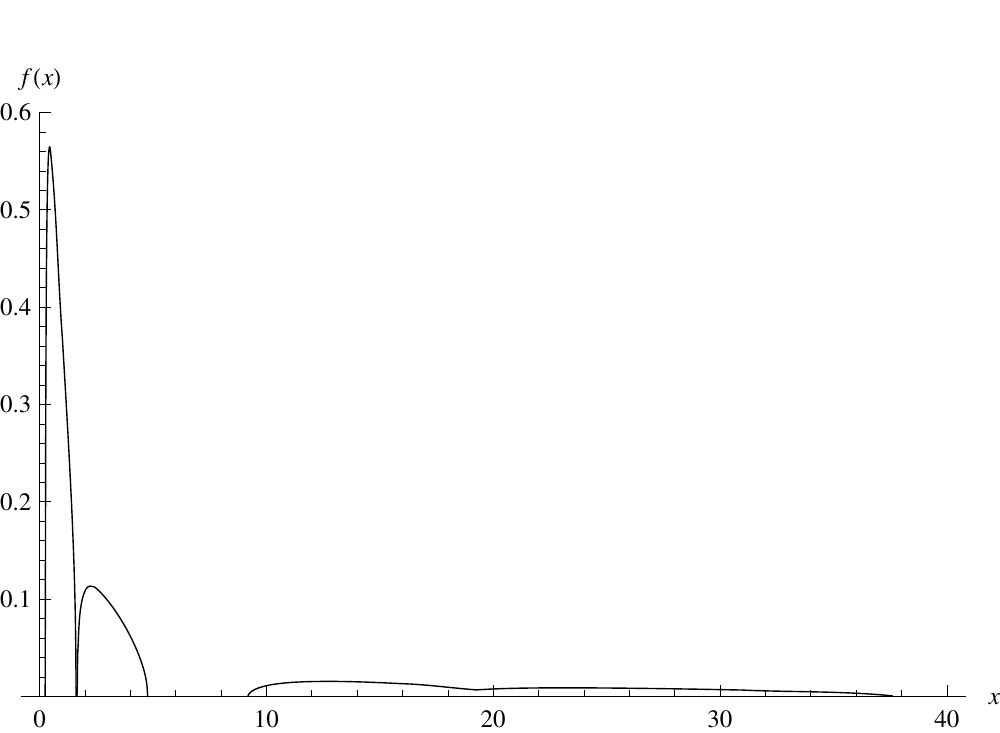}
\end{minipage}%
\begin{minipage}[t]{0.5\linewidth}
\includegraphics[width=2.5in]{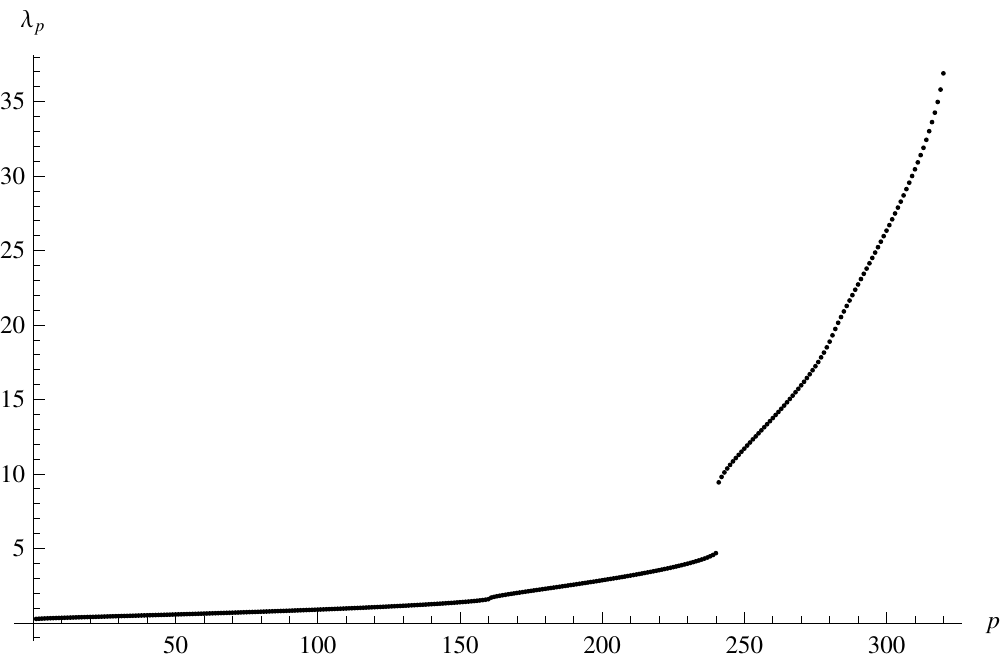}
\end{minipage}
\caption{The density curve of $F$ (left) and the average of the $i$-th ($i=1,\ldots,320$) sample eigenvalues (right) from 1000 replications for $H=0.5\delta_1+0.25\delta_3+0.125\delta_{15}+0.125\delta_{25}$ and $c=0.32$.}
\label{fig5}
\end{center}
\end{figure}

\begin{table}
\begin{center}
\caption{Frequency of estimates for the partition of $H$: $H_1=0.5\delta_1+0.25\delta_3$, $H_2=0.125\delta_{15}+0.125\delta_{25}$ with $p/n=0.32$.}
\label{table5}
\begin{tabular}{|lccc|}
\hline
Dimensions&{\bf k}=$(2,2)'$&{\bf k}=$(1,3)'$&{\bf k}=$(3,1)'$\\
\hline
$(p,n)=(320,1000)$&1000&0 &0\\
$(p,n)=(160,500)$ &1000&0 &0\\
$(p,n)=(64,200)$  &984 &0 &16\\
$(p,n)=(32,100)$  &911 &0 &89\\
$(p,n)=(16,50)$   &865 &0 &135\\
\hline
\end{tabular}
\end{center}
\end{table}

\begin{table}
\begin{center}
\caption{Estimates for $H=0.5\delta_1+0.25\delta_3+0.125\delta_{15}+0.125\delta_{25}$ with $p=320$ and $n=1000$.}
\label{table6}
\begin{tabular}{|llccccc|}
\hline
&&$a_1 $&$w_1$&$a_2$&$w_2$&$a_3$\\
\hline
GLME      &Mean   &1.0015&0.5015&3.0089&0.2485& 15.0133\\
          &St. D. &0.0080&0.0043&0.0270&0.0043&\ 0.2243\\
LME       &Mean   &1.0003&-     &2.9996&-     & 15.0061\\
          &St. D. &0.0042&-     &0.0165&-     &\ 0.2267\\
\hline
&&$w_3$&$a_4$&$w_4$&$d$&Time\\
\hline
GLME      &Mean   &0.1265& 25.1109&0.1235&0.1188&0.817s\\
          &St. D. &0.0058&\ 0.3361&0.0058&0.0639&\\
LME       &Mean   &0.1262& 25.1058&0.1238&0.1074&0.820s\\
          &St. D. &0.0058&\ 0.3428&0.0058&0.0641&\\
\hline
\end{tabular}
\end{center}
\end{table}

\par

\vskip 0.4cm
\setcounter{chapter}{5}
\setcounter{equation}{0} 
\noindent {\bf 5. Conclusions and remarks}

This paper investigates the problem of estimating the population spectral distribution from the sample eigenvalues in large dimensional framework. A local moment estimation procedure is proposed, by considering the division of a discrete PSD $H$ according to the separation of the LSD $F$. The new estimates are easy to compute and are proved to be consistent.

Our estimation procedure can be seen as an extension of the method in \cite{M08a}. The extension mainly focus on two aspects: first, the asymptotic clusters of sample eigenvalues generated by different population eigenvalues are not necessarily separate, that is, we drop the splitting condition; second, we don't need to know the weight parameters beforehand. These improvements enable our approach to be applied successfully to more complex PSDs.

At last, the proposed method is more efficient than that in \cite{Bai2010}. This could be attributed to two facts: our estimator uses much lower moments of the PSD $H$ (the highest order of the moments is 2$\max k_i-1$ used in the LME while it is $2k-1$ used in the BCY); moreover, our estimator is localized, then more efficient by removing possible mixture effect brought by sample eigenvalues from different $H_i$'s.

\vskip 0.4cm
\renewcommand{\theequation}{A.\arabic{equation}}
\renewcommand{\thefigure}{A.\arabic{figure}}
\setcounter{chapter}{5}
\setcounter{figure}{0}
\setcounter{equation}{0} 
\noindent {\large\bf Appendix: Calculation of the contour integrals in Equation \eqref{eq3.3}}

The possible poles in \eqref{eq3.3} are sample eigenvalues and zeros of $\underline{s}_n(u)$ on the real line. Thus, the next step is to determine which poles fall within the $i$-th integration region $C_i$.

Let $v=\min\{p,n\}$ and $\lambda_1<\cdots<\lambda_v$ be the nonzero sample eigenvalues. According to the main theorems in \cite{BS99}, these sample eigenvalues should form $m$ separate clusters for all large $p$ and $n$. Thus, with probability one, the $i$-th cluster of sample eigenvalues, denoted by $A_i$, falls within $C_i$ for all large $p$ and $n$.

On the other hand, notice that $\underline{s}_n(u)=0$ is equivalent to $\sum_{i=1}^v\lambda_i/(\lambda_i-u)=n$ (except for $p/n=1$, where the second equation would have an additional zero solution). Let $\mu_1<\cdots<\mu_v$ be zeros of $\underline{s}_n(u)$ (define $\mu_1=0$ if $p/n=1$), we have then
$$\mu_1<\lambda_1<\mu_2\cdots<\mu_v<\lambda_v.$$
Let $B_i=\{\mu_i: \mu_i\neq0,\ \lambda_i\in A_i\}$ ($i=1,\ldots,m$). From the proof of Lemma 1 in \cite{M08a}, we know that, with probability one, $B_i$ falls within $C_i$ for all large $p$ and $n$. A representation of $A_i$'s, $B_i$'s, and $C_i$'s is shown in Figure \ref{fig2} for a simple case. In order to differentiate between  $A_i$'s and $B_i$'s, the elements of $A_i$'s are plotted on the line $y=0.05$ and those of $B_i$'s are plotted on the line $y=-0.05$.

\begin{figure}
\centering
\includegraphics[width=5in,height=1.5in]{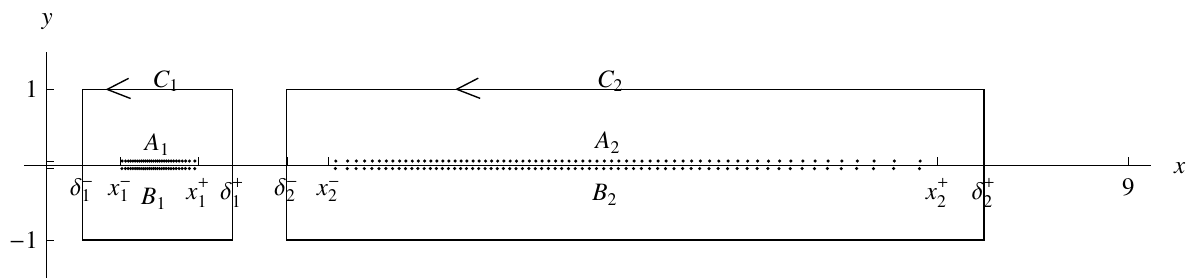}
\caption{Representation of $A_i$, $B_i$, and $C_i$ ($i=1,2$) where $H=0.3\delta_1+0.4\delta_4+0.3\delta_5$, $(c,p,n)=(0.1, 100,1000)$, and  $S_F=[x_1^-,x_1^+]\cup[x_2^-,x_2^+]=[0.6127, 1.2632]\cup[2.3484, 7.4137]$.}
\label{fig2}
\end{figure}

Therefore, the contour integral in \eqref{eq3.3} is formulated (approximately) as
\begin{equation}\label{eq5.1}
\frac{1}{{2\pi \rm i}}\oint_{C_i}\frac{z\underline{s}_n^\prime(z)}{\underline{s}_n^l(z)}dz
=\sum_{\lambda\in A_i}{\rm Res}(f_{ln}, \lambda)+\sum_{\mu\in B_i}{\rm Res}(f_{ln}, \mu),
\end{equation}
where $f_{ln}(z)=z\underline{s}_n^\prime(z)/\underline{s}_n^l(z)$.
The residues in \eqref{eq5.1} can be obtained by some elementary calculations. Residues from $A_i$ are simple:
$${\rm Res}(f_{ln}, \lambda)=-\lambda I(l=1).$$
Residues from $B_i$ are listed below for $l=1,\ldots,5$:
$$
{\rm Res}(f_{ln}, \mu)=
\begin{cases}
\mu &(l=1),\\
\frac{1}{\underline{s}_n^\prime(\mu)} &(l=2),\\
-\frac{\underline{s}_n^{\prime\prime}(\mu)}{2(\underline{s}_n^\prime(\mu))^3}&(l=3),\\
\frac{3(\underline{s}_n^{\prime\prime}(\mu))^2-\underline{s}_n^\prime(\mu)\underline{s}_n^{\prime\prime\prime}(\mu)}{6(\underline{s}_n^\prime(\mu))^5}&(l=4),\\
-\frac{15(\underline{s}_n^{\prime\prime}(\mu))^3-10\underline{s}_n^\prime(\mu)\underline{s}_n^{\prime\prime}(\mu)
\underline{s}_n^{\prime\prime\prime}(\mu)+(\underline{s}_n^\prime(\mu))^2\underline{s}_n^{(4)}(\mu)}{24(\underline{s}_n^\prime(\mu))^7}&(l=5).
\end{cases}
$$
For larger order $l$, we may get an analytic expression of ${\rm Res}(f_{ln}, \mu)$ from the following Mathematica code (here, the order of the moment is set to be 3):

\lstset{language=Mathematica}
\begin{lstlisting}
k = 3; * input the order of moment *
f = (z-mu)^k*z*D[sn[z],z]/(sn[z])^k;
D[f,{z,k-1}];
D[%*sn[z]^(2k-1),{z,2k-1}]/.z->mu;
D[sn[z],z]^(2k-1)(k-1)!(2k-1)!/.z->mu;
Simplify[%%/%,sn[mu]==0]
\end{lstlisting}

\vskip 0.4cm
\noindent {\large\bf Acknowledgment}

The authors would like to thank M. Yao Jianfeng (Telecom. ParisTech, Paris) for insightful discussions on the relationship \eqref{eq3.3} between moments of a PSD and the contour integrals. They are deeply grateful to the referees whose careful and detailed comments have led to many improvements of the manuscript.
\par

\vskip 0.4cm
\noindent{\large\bf References}
\begin{description}
\bibitem[Bai et~al.(2010)]{Bai2010}
\textsc{Bai, Z. D., Chen, J. Q. and Yao, J. F.} (2010).
On estimation of the population spectral distribution from a
high-dimensional sample covariance matrix.
\textit{Aust. N. Z. J. Stat.}
\textbf{52}, 423--437.

\bibitem[Bai and Silverstein(1998)]{BS98}
\textsc{Bai, Z. D. and Silverstein, J.} (1998).
No eigenvalues outside the support of the
limiting spectral distribution of large dimensional sample covariance
matrices.
\textit{ Ann. Probab.}
\textbf{26}, 316--345.

\bibitem[Bai and Silverstein(1999)]{BS99}
\textsc{Bai, Z. D. and Silverstein, J.} (1999).
Exact separation of eigenvalues of large dimensional
sample covariance matrices.
\textit{Ann. Probab.}
\textbf{27}, 1536--1555.


%

\bibitem[Chen et~al.(2011)]{CBY11}
\textsc{Chen, J.~Q., Delyon, B. and Yao, J.~F.} (2011).
On a Model Selection Problem from High-Dimensional Sample Covariance Matrices.
\textit{J. Multivariate Anal.}
\textbf{102}, 1388--1398.

\bibitem[El~Karoui(2008)]{KarE08}
\textsc{El Karoui, N.} (2008).
Spectrum estimation for large dimensional covariance matrices using
random matrix theory.
\textit{Ann. Statist.}
\textbf{36}, 2757--2790.

\bibitem[Hachem~et al.(2011)]{Hachem11}
\textsc{Hachem, W., Loubaton, P., Mestre, X., Najim, J. and Vallet, P.} (2011).
Large information plus noise random matrix models and
consistent subspace estimation in large sensor networks.
ArXiv:1106.5119v1.

\bibitem[Johnstone(2001)]{Johnstone01}
\textsc{Johnstone I.} (2001).
On the Distribution of the Largest Eigenvalue in Principal Components Analysis.
\textit{Ann. Statist.}
\textbf{29}, 295--327.

\bibitem[Li et~al.(2012)]{Li11}
\textsc{Li, W. M., Chen, J. Q., Qin, Y. L., Yao, J. F. and Bai, Z. D.} (2012).
Estimation of the population spectral distribution from a large dimensional sample covariance matrix.
\textit{Submitted.}

\bibitem[Mar{\v{c}}enko and Pastur(1967)]{MP67}
\textsc{Mar{\v{c}}enko, V. A. and Pastur, L. A.} (1967).
Distribution of eigenvalues in certain sets of random matrices.
\textit{Mat. Sb. (N.S.)}
\textbf{72}, 507--536.

\bibitem[Mestre(2008)]{M08a}
{\sc Mestre, X.} (2008).
\newblock Improved estimation of eigenvalues and eigenvectors
of covariance matrices using their sample estimates.
\newblock {\em IEEE Trans. Inform. Theory} {\bf 54}, 5113--5129.

\bibitem[Rao et al.(2008)]{RaoJ08}
\textsc{Rao, N. R., Mingo, J. A., Speicher, R. and Edelman, A.} (2008).
Statistical eigen-inference from large Wishart matrices.
\textit{Ann. Statist.}
\textbf{36}, 2850--2885.

\bibitem[Silverstein(1995)]{Silverstein95}
\textsc{Silverstein, J. W.} (1995).
Strong convergence of the empirical distribution of eigenvalues of
large-dimensional random matrices.
\textit{J. Multivariate Anal.}
\textbf{55}, 331--339.

\bibitem[Silverstein and Bai(1995)]{SilversteinB95}
\textsc{Silverstein, J. W. and Bai, Z. D.} (1995).
On the empirical distribution of eigenvalues of a class of
large-dimensional random matrices.
\textit{J. Multivariate Anal.}
\textbf{54}, 175--192.

\bibitem[Silverstein and Choi(1995)]{SilversteinC95}
\textsc{Silverstein, J. W. and Choi, S. I.} (1995).
Analysis of the limiting spectral distribution of large-dimensional
random matrices.
\textit{J. Multivariate Anal.}
\textbf{54}, 295--309.

%
%

\end{description}



\end{document}